\documentclass[11pt,a4paper]{article}
\usepackage[latin1]{inputenc}
\usepackage{amsmath}
\usepackage{amsthm}
\usepackage{amsfonts}	
\usepackage{amssymb}
\usepackage[pdftex]{color, graphicx}
\usepackage{rotating}

\usepackage{natbib}
\usepackage{cite}
\bibpunct{(}{)}{;}{a}{}{,}

\newtheorem{mydef}{Definition}
\newtheorem{myprop}{Proposition}
\newtheorem{corollary}{Corollary}

\usepackage{wasysym}
\usepackage{titlesec}
\usepackage[mathlines]{lineno}
\usepackage{setspace}
\usepackage{mathtools}

\usepackage[left=3cm,right=3cm,top=4cm,bottom=4cm]{geometry}

\titleformat*{\section}{\large\bfseries}

\makeatletter
\renewcommand*{\@fnsymbol}[1]{\ifcase#1\or* \else\@arabic{\numexpr#1-1\relax}\fi}
\makeatother

\begin{document}

\title{Assortment and the evolution of cooperation in a Moran process with exponential fitness\footnote{We are grateful to Martin Nowak for helpful comments.}}

\author{Daniel Cooney\thanks{Program in Applied and Computational Mathematics, Princeton University, New Jersey, U.S.A.} \and Benjamin Allen\thanks{Department of Mathematics, Emmanuel College, Massachusetts, U.S.A.} $^{, 4}$ \and Carl Veller\thanks{Department of Organismic and Evolutionary Biology, Harvard University, Massachusetts, U.S.A.} $^,$\thanks{Program for Evolutionary Dynamics, Harvard University, Massachusetts, U.S.A.} $^,$\thanks{carlveller@fas.harvard.edu}}\label{firstpage}

\date{}

\maketitle

\noindent \textbf{Abstract:} We study the evolution of cooperation in a finite population interacting according to a simple model of like-with-like assortment. Evolution proceeds as a Moran process, and payoffs from the underlying cooperator-defector game are translated to positive fitnesses by an exponential transformation. These evolutionary dynamics can arise, for example, in a nest-structured population with rare migration. The use of the exponential transformation, rather than the usual linear one, is appropriate when interactions have multiplicative fitness effects, and allows for a tractable characterization of the effect of assortment on the evolution of cooperation. We define two senses in which a greater degree of assortment can favour the evolution of cooperation, the first stronger than the second: (i) greater assortment increases, at all population states, the probability that the number of cooperators increases, relative to the probability that the number of defectors increases; and (ii) greater assortment increases the fixation probability of cooperation, relative to that of defection. We show that, by the stronger definition, greater assortment favours the evolution of cooperation for a subset of cooperative dilemmas: prisoners' dilemmas, snowdrift games, stag-hunt games, and some prisoners' delight games. For other cooperative dilemmas, greater assortment favours cooperation by the weak definition, but not by the strong definition. Our results hold for any strength of selection.\\


\newpage
\doublespacing


\section{Introduction}

Whether like-with-like assortment favours the evolution of cooperation is a critical question in social evolution and population biology more generally. Assortment can result if interactions are often between related individuals (e.g., if the population is `viscous'), or if individuals prefer to interact with those to whom they are similar (homophily). Any theoretical analysis of the relationship between assortment and cooperation requires the synthesis of two models: one for cooperative interactions, and one for population assortment.

The simplest models of cooperation are games between two players, where each player has available to it the same two strategies: cooperate ($C$) and defect ($D$). The paradigmatic example is the prisoners' dilemma \citep{rapoport1965, axelrod1984, nowak2006a}, but many other two-strategy games can be interpreted as cooperative dilemmas as well \citep{nowak2012}.

A simple model of assortment in two-strategy games was developed by \citet{grafen1979}. In his model, individuals can be of two types, corresponding to the two strategies in the underlying game. In a given time-step, an individual interacts with a same-type individual with probability $r$, and a random member of the population (possibly same-type, possibly not) with probability $1-r$. Higher values of $r$ correspond to greater degrees of assortment in the population.  This model has been extended in many ways \citep{eshel1982,bergstrom2003algebra,jansen2006altruism,van2012,alger2013homo,vancleve2014,allen2015}

Does increased assortment tend to favour cooperation? Of course, this question is not new, and has been approached from numerous angles. Most analyses have been in a deterministic, infinite-population setting. In such a setting, \citet[App.~A]{hamilton1971} determines, for several particular games of interest, whether selection is in favour of cooperators or defectors in a population with an interaction structure identical to that in \citet{grafen1979}. Extending this, \citet{eshel1982} analyze the effect of Grafen's assortment parameter on the conditions for evolutionary stability of cooperation and defection. They show that, for cooperative dilemmas broadly defined (two cooperators do better than two defectors, and a defector does better than a cooperator when faced with a cooperator), there are threshold assortment levels $0 < r^{**} < r^* < 1$ such that, if $r>r^*$, cooperation is the only ESS, but if $r<r^{**}$, defection is an ESS (and the only one in the case of a prisoners' dilemma). That is, sufficiently high relatedness renders cooperation evolutionarily stable, while sufficiently low relatedness renders defection evolutionarily stable. 

More recently, the effect of assortment on the evolution of cooperation has been studied in the context of finite populations subject to stochastic evolution. In this setting, a natural measure of the success of a strategy is the weight of its pure state (where the population is monomorphic for that strategy) in the stationary distribution of the evolutionary process \citep{fudenberg2006, allen2014, veller2016}. In a two-strategy game (such as the usual cooperative dilemmas) with symmetric, small mutation rates, the superiority of one strategy over the other corresponds to its having a higher fixation probability \citep{fudenberg2006, allen2014}. In a very general model of social evolution that allows for arbitrary dependence of payoffs on strategy frequencies and population structure of any sort, \citet{vancleve2015} derives, in the limit of weak selection, the conditions under which the fixation probability of cooperation exceeds that of defection. 

An alternative way to study whether assortment aids cooperation is to determine the effect on key parameters of the evolutionary process, positive or negative, of marginal changes in the assortment parameter. In the economics literature, this general approach is called `comparative statics': rather than trying to determine the \emph{exact} effect a change in one variable has on another, we instead try to determine just the \emph{direction} of the effect (its `comparative static'). Many empirical studies of social evolution test the broad `comparative statics' statement that increased relatedness should be associated with more cooperation, and much of social evolutionary theory is interpreted in this way \citep{frank2013}. It is therefore important that the comparative statics of social evolution be studied extensively from a theoretical perspective.

A step in this direction was recently provided by \citet{allen2015}, who studied the effect of assortment on the evolution of cooperation in a finite population subject to Wright-Fisher evolution. Interaction each period accords with Grafen's assortment process, with expected payoffs derived from an underlying two-strategy game. Expected payoffs are translated to fitnesses according to a linear mapping mediated by a selection strength parameter $w$, and the next period's population composition is chosen stochastically on the basis of these fitnesses. \citet{allen2015} derive analytical results for the `weak selection limit', $w \to 0$. In this limit, for most cooperative dilemmas, increasing relatedness favours cooperation in the sense that it increases the fixation probability of a single cooperator relative to that of a single defector.

Here, we make use of a relatively recent advance in the study of finite-population evolutionary game theory to give a more general characterization of the comparative statics of assortment in Grafen's model. We model evolution as a Moran process \citep{moran1958}, occurring in a finite population that interacts according to Grafen's assortment model. Expected payoffs translate to fitnesses according to an \emph{exponential} mapping mediated by a selection strength parameter $w$ \citep{traulsen2008}. With this setup, we can derive a more general result than \citet{allen2015}: for \emph{any} selection strength $w$, increasing $r$ increases the fixation probability of a cooperator relative to that of a defector for the same set of cooperative dilemmas to which their result applies. We also define a stronger sense in which increased $r$ favours cooperation, and show that this too holds for a natural class of cooperative dilemmas: prisoners' dilemmas, snowdrift games, stag-hunt games, and some prisoners' delight games. To distinguish this stronger definition, we provide an example of a cooperative dilemma where assortment favours cooperation by the weaker definition of \citet{allen2015}, but not by our stronger definition.

For most of this work, we follow in the footsteps of many others \citep{grafen1979, eshel1982,bergstrom2003algebra,jansen2006altruism,alger2013homo, vancleve2014,okasha2015hamilton} by treating $r$ as an abstract assortment parameter without specifying how this assortment arises mechanistically.  In Section \ref{sec:biologicalmodel}, however, we illustrate how our model could apply to a nest-structured population.  In this setting, the parameter $r$ quantifies the frequency of interaction among nest-mates.  

Our results offer insights into the effect of assortment on the evolution of cooperation, and highlight the usefulness of the exponential fitness function in evolutionary game theory.

\section{The general model}

Social interactions are between two individuals, with the payoffs from any one interaction defined by the payoff matrix
\[\left.
\begin{array}{c|cc}
  & C & D\\
  \hline
C & R & S\\
D & T & P
\end{array} \right.\]
Following \citet{nowak2012}, we say that the game is a `cooperative dilemma' if $R>P$ (two cooperators do better than two defectors) and at least one of the following conditions holds: (i) $T>R$ (it is better to defect if one's opponent cooperates), (ii) $P>S$ (it is better to defect if one's opponent defects), (iii) $T>S$ (if one interactant cooperates and the other defects, then the defector does better than the cooperator). The game is a prisoners' dilemma if $T>R>P>S$; it is a snowdrift game if $T>R>S>P$; it is a stag-hunt game if $R>T\geq P > S$. 

The population is of constant size $N$. Each member's type is either cooperator (plays $C$) or defector (plays $D$). Each period, each member of the population receives its expected payoff from the following interaction: with probability $r$, it interacts with an individual of the same type, and with probability $1-r$, it interacts with a random member of the population, with each member equally likely (including the individual itself---the possibility of self-interaction simplifies the analysis).

The population state at a given time is defined by the number of cooperators, $i$. If $1\leq i \leq N-1$, then the expected payoff to a cooperator and a defector are:
\begin{linenomath}
\begin{eqnarray}
\pi^C_i &=& rR + (1-r)\frac{iR + (N-i)S}{N},\nonumber\\
\pi^D_i &=& rP + (1-r)\frac{iT + (N-i)P}{N}.\label{payoffs}
\end{eqnarray}
\end{linenomath}
These expected payoffs in turn translate to (positive) fitnesses according to some monotonic transformation.

Evolution is modelled as a Moran process \citep{moran1958, nowak2006}. Each period, the fitnesses of the cooperators and defectors in the population are calculated as above. One member is chosen for reproduction; the probability of being chosen is proportional to fitness. One member of the population is chosen to die; each member is equally likely to be chosen. A new individual is then born, taking the place of the one chosen to die. The new individual is of the same type as its parent, the member chosen for reproduction (which can be the same as the individual that was chosen to die).

Thus, if there are $i$ cooperators this period in the population, with $1 \leq i \leq N-1$, then denoting by $p_{i,j}$ the probability that there will be $j$ cooperators next period, we have:
\begin{linenomath}
\begin{eqnarray}
p_{i,i+1} &=& \frac{N-i}{N}\cdot\frac{if^C_i}{if^C_i + (N-i)f^D_i},\nonumber\\
p_{i,i-1} &=& \frac{i}{N}\cdot\frac{(N-i)f^D_i}{if^C_i + (N-i)f^D_i},\nonumber\\
p_{i,i} &=& 1 - p_{i,i+1} - p_{i,i-1}\label{transitions}. 
\end{eqnarray}
\end{linenomath}
We term $p_{i,i+1}/p_{i,i-1}$ the `relative rate of increase of cooperation' at population state $i$. This leads to our first definition by which higher assortment can favour cooperation: if it increases the relative rate of increase of cooperation for all interior $i$.

\begin{mydef}
Assortment favours cooperation in the \emph{strong sense} if $\frac{\partial}{\partial r} \frac{p_{i,i+1}}{p_{i,i-1}} > 0$ for all $i = 1, \ldots, N-1$.
\end{mydef}

Notice that, given the Moran process transition probabilities in \eqref{transitions}, $p_{i,i+1}/p_{i,i-1} = f^C_i/f^D_i$. So the condition for assortment to favour cooperation in the strong sense simplifies to $\frac{\partial}{\partial r} \frac{f^C_i}{f^D_i} > 0$ for all $i = 1, \ldots, N-1$.

Now consider the probability that a single cooperator, presently in a population that is otherwise all defectors, takes over the population; i.e., its `fixation probability', $\rho^C$. For our frequency-dependent Moran process, this is given by \citep{nowak2006}:
\[\rho^C = \frac{1}{1+\sum_{j=1}^{N-1}\prod_{i=1}^{j} \frac{f^D_i}{f^C_i}}.\]
The fixation probability of a single $D$ in a population otherwise all-$C$ is
\[\rho^D = \frac{\prod_{i=1}^{N-1}\frac{f^D_i}{f^C_i}}{1+\sum_{j=1}^{N-1}\prod_{i=1}^{j} \frac{f^D_i}{f^C_i}}.\]
From these expressions, the `relative fixation probability of cooperation', $\rho^C/\rho^D$, simplifies to:
\[\frac{\rho^C}{\rho^D} = \prod_{i=1}^{N-1}\frac{f^C_i}{f^D_i}.\]
This leads to our second definition by which higher assortment can favour cooperation: if it increases the relative fixation probability of cooperation.
\begin{mydef}
Assortment favours cooperation in the \emph{weak sense} if $\frac{\partial}{\partial r} \frac{\rho^C}{\rho^D} > 0$.
\end{mydef}
The use of the terminology `weak' and `strong' for these definitions is justified by the fact that the strong definition implies the weak definition: if $\frac{\partial}{\partial r} \frac{f^C_i}{f^D_i} > 0$ for all $i = 1, \ldots, N-1$, then clearly $\frac{\partial}{\partial r} \prod_{i=1}^{N-1}\frac{f^C_i}{f^D_i} > 0$.

\section{The exponential fitness transformation}

Following \citet{blume1993} and \citet{traulsen2008}, we use an exponential transformation of expected payoffs to (positive) fitnesses:
\begin{linenomath}
\begin{eqnarray}
f^C_i = \exp{\left(w \pi^C_i\right)},\; f^D_i = \exp{\left(w \pi^D_i\right)},\label{exponential} 
\end{eqnarray}
\end{linenomath}
where $w$ is the strength of selection. We should note, as do \citet{traulsen2008}, that this choice of transformation is for the most part arbitrary. Any transformation should translate payoffs to fitnesses monotonically, and in such a way that the fitnesses are positive (since transition probabilities in the stochastic evolutionary process are proportional to them). In the common linear transformation, $f = 1+w\pi$, when some payoffs are negative, the selection strength needs to be sufficiently small to ensure that all fitnesses are positive. In contrast, the exponential transformation that we use guarantees that all fitnesses are positive, even if some payoffs in the underlying game are negative.

Moreover, the exponential transformation is more conducive to a `multiplicative fitness effects' interpretation. If an individual has a total of $K$ interactions, with each interaction $k$ having a small fitness consequence $\exp{\left(w\pi^k/K\right)}\approx 1+w\pi^k/K$, and if these fitness consequences combine multiplicatively, then the total fitness consequence is
\[f = \prod_{k=1}^{K}\exp{(w\pi^k/K)} = \exp{\left(w\frac{1}{K}\sum_{k=1}^K \pi^k\right)}.\]
If the $\pi^k$ are interpreted as payoffs from an underlying game, then the expression above can be read as $\exp{\left(w\pi\right)}$, where $\pi$ is the average payoff to the individual. Treating averages as expectations, this is the same as the expression in \eqref{exponential}.

Finally, we note that in the weak selection limit $w \to 0$, the exponential fitness transformation \eqref{exponential} coincides with the usual linear fitness transformation $f = 1+w\pi$, since Taylor expanding \eqref{exponential} around $w=0$ gives $\exp{\left(w \pi\right)} = 1 + w \pi + \mathcal{O}(w^2)$. Thus, models that use the linear fitness transformation and the weak selection limit, such as in \citet{allen2015} and \citet{nowak2004}, can be thought of as special cases of the exponential fitness formulation.

\section{Results for exponential fitness}

We now consider the effect of assortment on the evolution of cooperation when fitnesses are calculated according to \eqref{exponential}. The relative rate of increase of cooperation is calculated by successively applying \eqref{transitions} and \eqref{exponential}:
\begin{linenomath}
\begin{eqnarray}
p_{i,i+1}/p_{i,i-1} = f^C_i/f^D_i = \exp{\left(w\left[\pi_i^C - \pi_i^D\right]\right)},\label{rate}
\end{eqnarray}
\end{linenomath}
with $\pi_i^C$ and $\pi_i^D$ as specified in \eqref{payoffs}.

This allows us to analyze for which games assortment favours cooperation in the strong sense of Definition 1 (Fig.~1):
\begin{myprop}
Assortment favours cooperation in the strong sense if $R\geq S$ and $T\geq P$, with at least one inequality strict.
\end{myprop}
\begin{proof}
We have from \eqref{rate} that
\[\frac{\partial}{\partial r}\frac{p_{i,i+1}}{p_{i,i-1}} = \exp{\left(w\left[\pi_i^C - \pi_i^D\right]\right)}w\frac{\partial}{\partial r}\left[\pi_i^C - \pi_i^D\right],\]
and since the exponential is positive, the sign of $\frac{\partial}{\partial r}\frac{p_{i,i+1}}{p_{i,i-1}}$ is the same as the sign of $\frac{\partial}{\partial r}\left[\pi_i^C - \pi_i^D\right]$. Writing $x := i/N$, so that $x\in(0,1)$, we have from \eqref{payoffs} that this quantity is given by
\begin{linenomath}
\begin{eqnarray}
\frac{\partial}{\partial r}\pi_i^C - \frac{\partial}{\partial r}\pi_i^D &=& [R - xR - (1-x)S] - [P - xT - (1-x)P]\nonumber\\
&=& (1-x)(R-S) + x(T-P).\label{strict_condition}
\end{eqnarray}
\end{linenomath}
If both $R\geq S$ and $T\geq P$, with at least one inequality strict, then \eqref{strict_condition} is strictly positive for all $x\in (0,1)$, and therefore for all $i$ such that $1\leq i \leq N-1$. So, therefore, is $\frac{\partial}{\partial r}\frac{p_{i,i+1}}{p_{i,i-1}}$.
\end{proof}

Notice that the condition $R\geq S$ and $T\geq P$ provides a natural stringent definition of a cooperative game: No matter what my strategy is, I would prefer that my opponent cooperate. This condition is satisfied, in particular, by prisoners' dilemmas ($T > R > P > S$), snowdrift games ($T > R > S > P$), and stag-hunt games ($R > T \geq P > S$). In fact, the only other cooperative dilemmas -- as defined by \citet{nowak2012} and in Section 2 -- that satisfy this condition are those where $R > T > S > P$; i.e., a subset of ``prisoners' delight'' games \citep{binmore2004, skyrms2014}.

We can also prove a partial converse to Proposition 1, which shows that Proposition 1 can be strengthened to an `if, and only if' statement provided the population size $N$ is sufficiently large.
\begin{myprop}
If assortment favours cooperation in the strong sense, then one of the following three sets of conditions is satisfied: 
\begin{itemize}
 \item[(2a)] $R \geq S$ and $T \geq P$, with at most one inequality strict.
 \item[(2b)] $R > S$, $T < P$, and $N < \frac{R-S+P-T}{P-T}$.
 \item[(2c)] $R < S$, $T > P$, and $N < \frac{S-R+T-P}{S-R}$.
 \end{itemize}
\end{myprop}
\begin{proof}
See Appendix A.
\end{proof}

\begin{figure}[t]
\includegraphics[width=7cm]{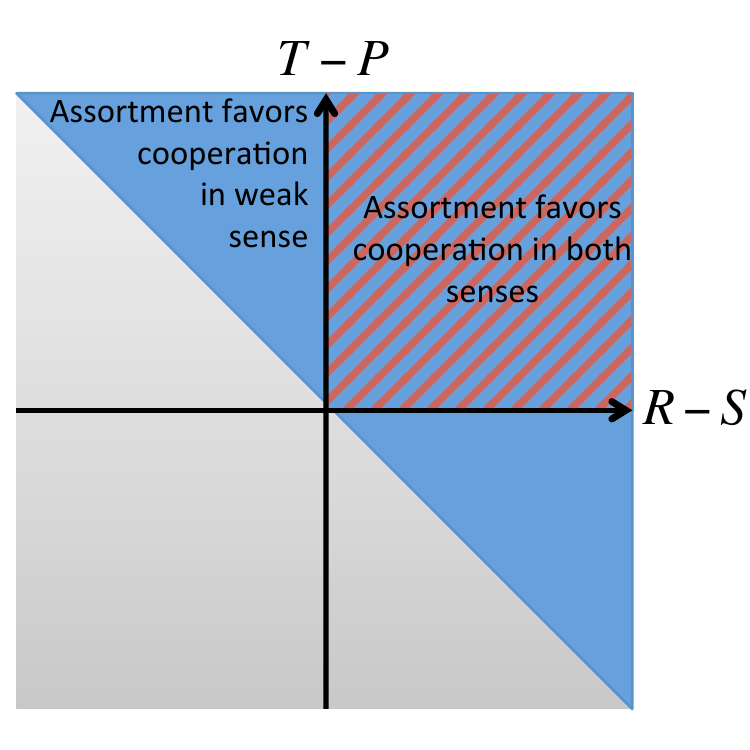}
\centering
\caption{A summary of our results for two-player games. Assortment favours cooperation in the \emph{strong sense} (increases the relative rate of increase of cooperators for all intermediate population states) if $R\geq S$ and $T\geq P$, with at least one inequality strict (blue and red cross-hatch). For large enough population sizes, this is a necessary condition as well. Assortment favours cooperation in the \emph{weak sense} (increases the fixation probability of a single cooperator, relative to a single defector) if, and only if, $R-S > P-T$ (blue).}
\end{figure}

We now determine how much broader the conditions are under which assortment favours cooperation in the weaker sense of Definition 2.

As noted by \citet{traulsen2008}, the ratio of fixation probabilities $\rho^C/\rho^D$ in the frequency-dependent Moran process simplifies significantly with the exponential fitness transformation \eqref{exponential}:
\begin{linenomath}
\begin{eqnarray}
\frac{\rho^C}{\rho^D} &=& \prod_{i=1}^{N-1} \frac{f^C_i}{f^D_i} = \prod_{i=1}^{N-1} \exp{\left(w\left[\pi_i^C - \pi_i^D\right]\right)},\nonumber\\
&=& \exp{\left(w\sum_{i=1}^{N-1}\left[\pi_i^C - \pi_i^D\right]\right)}.\label{fixation_ratio1}
\end{eqnarray}
\end{linenomath}
But in our case, we have from \eqref{payoffs} that for each $i$, $\pi_i^C + \pi_{N-i}^C = 2rR + (1-r)(R+S)$ and $\pi_i^D + \pi_{N-i}^D = 2rP + (1-r)(T+P)$, and if $N$ is even, then $\pi_{N/2}^C= rR + (1-r)(R+S)/2$ and $\pi_{N/2}^D = rP + (1-r)(T+P)/2$. Grouping payoff terms for $i$ and $N-i$ in \eqref{fixation_ratio1}, we see that, for both odd and even $N$,
\begin{linenomath}
\begin{eqnarray}
\frac{\rho^C}{\rho^D} &=& \exp{\left(w\frac{N-1}{2}[2rR + (1-r)(R+S) - 2rP - (1-r)(T+P)]\right)}\nonumber\\
&=& \exp{\left(w\frac{N-1}{2}[(R-P+S-T) + r(R-P-S+T)]\right)}.\label{fixation_ratio}
\end{eqnarray}
\end{linenomath}

This allows us to analyze for which games assortment favours cooperation in the weak sense of Definition 2 (Fig.~1).
\begin{myprop}
Assortment favours cooperation in the weak sense if, and only if, $R-S > P-T$.
\end{myprop}
\begin{proof}
From \eqref{fixation_ratio}, we have that 
\[\frac{\partial}{\partial r} \frac{\rho^C}{\rho^D} = \exp{\left(w\frac{N-1}{2}[(R-P+S-T) + r(R-P-S+T)]\right)}w\frac{N-1}{2}(R-P-S+T).\]
Since the exponential is positive, $\frac{\partial}{\partial r} \frac{\rho^C}{\rho^D} > 0 \; \Leftrightarrow \; R-P-S+T>0$.
\end{proof}

This condition corresponds to that found by \citet{allen2015} for assortment favouring cooperation in the same weak sense, but their result applies only in the weak selection limit ($w\to 0$), whereas our result holds for all selection strengths. 

A corollary of Proposition 3 follows from a result due to \citet{fudenberg2006}. Consider an evolving finite population of cooperators and defectors, where cooperators mutate to defectors and defectors to cooperators at the same rate. If this rate of mutation is very small (the `weak mutation limit'), the long-run distribution of the (ergodic) evolutionary process over all population states $i = 1, 2, \ldots, N$ is close to a distribution over just the two pure states all-$D$ ($i = 0$) and all-$C$ ($i = N$). This distribution is given by $\left(\frac{\rho^D}{\rho^D+\rho^C}, \frac{\rho^C}{\rho^D+\rho^C}\right)$, whose dependence on $\rho^C/\rho^D$ can explicitly be seen in the form $\left(\frac{1}{1+\rho^C/\rho^D}, 1-\frac{1}{1+\rho^C/\rho^D}\right)$. If greater assortment increases the ratio $\rho^C/\rho^D$ (i.e., if assortment favours cooperation in our weak sense), it also increases the weight on the cooperative state in this weak mutation limit.

\begin{corollary}
In the Moran process with symmetric mutation rates, in the weak-mutation limit, greater assortment increases the stationary distribution's weight on the cooperative state all-$C$ if, and only if, $R-P-S+T>0$.
\end{corollary}

If the condition $R-P-S+T>0$ holds and mutations are very infrequent, then this corollary can be interpreted as follows: at a time point far in the future, a population with a high degree of assortment is more likely to be all cooperating than is a population with a low degree of assortment. Alternatively, among many evolving populations, cooperation will be more common among those with high assortment than those with low assortment.

We have demonstrated that, for a large class of cooperative dilemmas, assortment favours cooperation in our strong sense, and that for an even broader class, assortment favours cooperation in the weak sense. To demonstrate the utility of the distinction between our weak and strong definitions, we provide an example of a game in which assortment favours cooperation in the weak sense but not the strong sense. 

Suppose we have a population of size $N > 10$ interacting according to Grafen's assortment model, with payoffs specified by the game 
\[\left.
\begin{array}{c|cc}
  & C & D\\
  \hline
C & R = 10 & S = 1\\
D & T = 8  & P = 9
\end{array} \right..\]
This is a `cooperative dilemma' in the sense of \citet{nowak2012} because two cooperators get a higher payoff than two defectors ($R>P$), it is better to be the defector in a cooperator-defector pair ($T>S$), and it is better to defect if one's opponent is a defector ($P>S$). 

Here, assortment favours cooperation in the weak sense, because $R-P-S+T = 8 > 0$. But assortment does not favour cooperation in the strong sense, because none of the conditions (2a), (2b), and (2c) holds: (2a) and (2c) fail because $T<P$; (2b) fails because $\frac{R-S+P-T}{P-T} = 10 < N$.

\section{Realization of the model in a nest-structured population}\label{sec:biologicalmodel}

So far we have described the model in abstract terms.  Here we show how the model can be realized in a nest-structured population, with $r$ describing the relative frequency of interaction among nest-mates.  

Consider a population with $N$ nests, each containing a single adult.  Each time-step corresponds to a single (non-overlapping) generation.  Each generation, each adult produces a large number of juveniles. Each juvenile interacts with a large number of others, a proportion $r$ of these being with juveniles from the same nest, and the remaining proportion $1-r$ being with juveniles chosen uniformly from the population at large. These interactions determine a juvenile's fitness multiplicatively, so that fitness is an exponential function of the sum of individual payoffs. 

After interaction, a new head for each nest is chosen, forming the next generation of adults.  With high probability ($1-\mu$, where $\mu \ll 1$ represents the migration rate) the new head of a nest is chosen from among the juveniles who were born at this nest.  Otherwise (with probability $\mu$) the new head is chosen from the juvenile population at large.  In each case, the choice is made proportionally to fitness.  

In the limit $\mu \to 0$, the dynamics of the above model coincide with those of the simple model we have studied in this paper (up to a rescaling of time). The reason is that, as the migration rate becomes very small, the probability of there being two migrants in a given generation becomes negligible compared to that of there being only one migrant. Thus at most one nest changes its strategy per generation, consistent with the Moran process.  

\section{Games with more than two players}

In reality, many cooperative dilemmas are not pairwise interactions, but involve more than two players. Here, we shall demonstrate that our methodology can extend to a particularly simple $n$-player cooperative dilemma, a linear public goods game. We shall show that increased assortment favours the evolution of cooperation in this case.

Interactions occur in groups of size $n$, where each interactant is either a cooperator ($C$) or a defector ($D$). Assortment is defined as $r = \mathbb{P}(C|C) - \mathbb{P}(C|D)$, where $\mathbb{P}(C|D)$ is the probability that a randomly chosen interactant of a defector is a cooperator, and $\mathbb{P}(C|C)$ is the probability that a randomly chosen interactant of a cooperator is a cooperator \citep{van2009group}. (It is easily seen that $\mathbb{P}(C|C) - \mathbb{P}(C|D) = \mathbb{P}(D|D) - \mathbb{P}(D|C)$, so that this definition does not depend on the labelling of strategies.) This definition is consistent with that used in our two-player setup above.

Assortment derives from a `population structure', which determines, from the number of cooperators in the population, their distribution among the interacting groups. Formally, in a population of size $N$, a population structure $m$ (for `matching function') is a vector-valued function that specifies, for each possible number $i$ of cooperators in the population, a probability distribution $(m_{[0]}(i), m_{[1]}(i), \ldots, m_{[n]}(i))$ over the $n+1$ possible group compositions, so that $m_{[j]}(i)$ is the probability that, when there are $i$ cooperators in the whole population, a randomly chosen group contains exactly $j$ cooperators. (The subscript in $m_{[j]}(i)$ is bracketed to distinguish that this is the number of cooperators in the group, while we have throughout used an unbracketed subscript to indicate the number of cooperators in the whole population.)

One complication arises. In an infinite-population setting \citep{van2011replicator, jensen2014evolutionary, rigos2016assortativity}, where the proportion $x\in [0,1]$ of cooperators is specified, it is required that the population structure $f$, applied to the proportion of cooperators $x$, recapitulate that proportion: $x = \sum_{j=0}^n j m_{[j]}(x)$. In our finite-population case, where the law of large numbers does not apply, $N/n$ groups drawn independently from the distribution $m(i)$ need not contain a total of $i$ cooperators. To sidestep this complication, we shall make an assumption analogous to that of self-interaction in the two-player game setup: interactants in a group are drawn \emph{with replacement} from the population, according to the distribution $m(i)$. That is, we shall assume that an individual can take on multiple positions in the group, so that a group of $n$ need not contain $n$ distinct individuals.

We shall be interested in the population structures that maintain, for all intermediate $i = 1, 2, \ldots, N-1$, a constant assortment parameter $r$. One such population structure, and a natural extension of the population structure we analyzed for two-player games, is this: A cooperator's $n-1$ group-mates are, with probability $r$, all cooperators, and with probability $1-r$ are drawn randomly and independently from the population (i.e., they constitute a binomial sample of size $n-1$). When there are no cooperators in the population (respectively, no defectors), then all groups comprise only defectors (respectively, only cooperators). The exact population structure that leads to a constant $r$ will not matter in what follows, because the game we shall consider is linear. For general (nonlinear) $n$-player games, however, the population structure will matter.

The game played by the $n$ constituents of a group is a linear, or additive, public goods game. A cooperator produces, at individual payoff cost $c$ to itself, a public good that provides benefit $b$ to each member of the group. A defector enjoys the benefits of the public good produced by cooperators in the group, but does not produce any itself. The payoff to a cooperator in a group with $j$ cooperators is $\pi^C_{(j)} = jb - c$, while the payoff to a defector in a group with $j$ cooperators is $\pi^C_{(j)} = jb$. 

The expected payoff to a cooperator and a defector are, respectively:
\begin{linenomath}
\begin{eqnarray}
\pi_i^C &=& -c + b + (n-1)\mathbb{P}(C|C)b;\nonumber\\
\pi_i^D &=& (n-1)\mathbb{P}(C|D)b.\label{multiplayer}
\end{eqnarray}
\end{linenomath}
With a constant-$r$ assortment process and this game, when there are $i$ cooperators in the population, we have \citep{van2011replicator}:
\begin{equation*}
\pi_i^C - \pi_i^D = (n-1)[\mathbb{P}(C|C) - \mathbb{P}(C|D)]b - c + b = (n-1)rb - c +b.
\end{equation*}

As before, we calculate fitness exponentially from expected payoffs, and update the population according to a Moran process. Checking our strong condition by which assortment favours cooperation, Definition 1:
\begin{equation*}
p_{i,i+1}/p_{i,i-1} = f_i^C/f_i^D = \exp\left(w[\pi_i^C - \pi_i^D]\right) = \exp\left(w[(n-1)rb - c + b]\right),
\end{equation*}
so that
\begin{equation*}
\frac{\partial}{\partial r} \frac{p_{i,i+1}}{p_{i,i-1}} = w(n-1)b\exp\left(w[(n-1)rb - c + b]\right) > 0.
\end{equation*}
Therefore, assortment favours the evolution of cooperation in the strong sense here.

The above results apply to a linear public goods game.  For nonlinear public goods games, a single assortment parameter does not suffice to characterize the effects of population structure on natural selection; rather, a full probability distribution over all possible group compositions is required \citep{van2009group, ohtsuki2014evolutionary}.  Therefore, while nonlinear public goods games are important from both a theoretical \citep{archetti2012review} and empirical \citep{GoreSnowdrift} perspective, an analysis of how assortment affects cooperation in such games is beyond the scope of the current work.  

\section{Discussion}

We have studied the effect of assortment on the evolution of cooperation, making use of a simple model of like-with-like assortment in two-strategy games. We have demonstrated that, for many cooperative games of interest (including prisoners' dilemmas, snowdrift games, and stag-hunt games), increased assortment favours the evolution of cooperation in two senses: it increases the rate at which the cooperative strategy increases in the population, and it increases the fixation probability of cooperators relative to that of defectors.

Most studies of stochastic game theory -- for example, \citet{allen2015} in a setting similar to ours -- assume selection to be weak. This assumption is made because standard stochastic models are otherwise intractable, with analytical results impossible. Here, making use of an exponential transformation of payoffs to fitnesses instead of the usual linear one, we have demonstrated that informative analytical results may be obtained in the more general case allowing any strength of selection. 

Moreover, this exponential transformation is not simply a mathematical expediency. We have argued that it is the appropriate transformation of average payoff to fitness when fitness derives from many interactions with multiplicative fitness effects. Multiplicative fitness effects are a common assumption in population genetics \citep{ewens2004}, but not in evolutionary game theory, where additive fitness effects are standard. 

In addition to their biological import, our results highlight the usefulness of this exponential payoff-fitness transformation in the Moran process [also see \citet{traulsen2008}], and suggest that tractable results could be obtained by using this as a model to answer other questions in social evolution and evolutionary game theory more broadly.

\bibliographystyle{abbrvnat}
\bibliography{grafen_moran_cite}

\section*{Appendix A: Proof of Proposition 2}

We prove Proposition 2 by considering all possible cases for the signs of $R-S$ and $T-P$. 

\begin{itemize}

\item From Proposition 3, we know that if $P \geq T$ and $R \geq S$ (with at least one inequality strict), $\frac{\partial}{\partial r}\frac{p_{i,i+1}}{p_{i,i-1}} > 0$. 

\item If $R \leq S$ and $T \leq P$, both terms in \eqref{strict_condition} are non-positive, and so $\frac{\partial}{\partial r}\frac{p_{i,i+1}}{p_{i,i-1}} \leq 0$.

\item If $R < S$ and $T > P$, the term $R-S$ in \eqref{strict_condition} is negative, and the term $T-P$ is positive. The expression \eqref{strict_condition} is therefore increasing in $x:=i/N$, and therefore smallest when $i = 1$. So, in this case, $\frac{\partial}{\partial r}\frac{p_{i,i+1}}{p_{i,i-1}} > 0$ for each $i = 1, \ldots, N-1$ if, and only if, \eqref{strict_condition} is positive when $i=1$. When $i=1$, \eqref{strict_condition} equals $[(N-1)(R-S) + (T-P)]/N$, which is positive if, and only if, $N < (S-R + T-P)/(S-R)$.

\item If $R>S$ and $T<P$, the term $R-S$ in \eqref{strict_condition} is positive, and the term $T-P$ is negative. The expression \eqref{strict_condition} is therefore decreasing in $x:=i/N$, and therefore smallest when $i = (N-1)/N$. So, in this case, $\frac{\partial}{\partial r}\frac{p_{i,i+1}}{p_{i,i-1}} > 0$ for each $i = 1, \ldots, N-1$ if, and only if, \eqref{strict_condition} is positive when $i=(N-1)/N$. When $i=(N-1)/N$, \eqref{strict_condition} equals $[(R-S) + (N-1)(T-P)]/N$, which is positive if, and only if, $N < (R-S + P-T)/(P-T)$. 

\end{itemize}

\end{document}